\documentclass[11pt,a4paper,reqno]{article}
\usepackage{comment}
\usepackage[utf8]{inputenc}
\usepackage[T1]{fontenc}
\usepackage{mathtools, blkarray, booktabs}
\usepackage{amssymb}
\usepackage{mathtools}
\usepackage{latexsym}
\usepackage{amsmath}
\usepackage{graphicx}
\usepackage{relsize}
\usepackage{amsthm}
\usepackage{empheq}
\usepackage{wasysym}
\usepackage{mathrsfs}

\usepackage{bm}
\usepackage{booktabs}
\usepackage[dvipsnames]{xcolor}
\usepackage{tikz,pgfplots,lipsum,lmodern}
\usepackage{enumitem}
\usepackage[most]{tcolorbox}

\newcommand{\diagdots}[3][-25]{%
  \rotatebox{#1}{\makebox[0pt]{\makebox[#2]{\xleaders\hbox{$\cdot$\hskip#3}\hfill\kern0pt}}}%
}

\usepackage[margin=2.5cm]{geometry}

 \usepackage[section]{placeins}

\numberwithin{equation}{section}

\theoremstyle{definition}
\newtheorem{theorem}{Theorem}[section]

\newtheorem{proposition}[theorem]{Proposition}
\newtheorem{definition}[theorem]{Definition}
\newtheorem{example}[theorem]{Example}
\newtheorem{notation}[theorem]{Notation}
\newtheorem{remark}[theorem]{Remark}
\newtheorem{lemma}[theorem]{Lemma}

\usepackage{calrsfs}

\usepackage{calrsfs}
\DeclareMathAlphabet{\pazocal}{OMS}{zplm}{m}{n}

\newcommand{\numberset}{\mathbb}

\newcommand{\F}{\numberset{F}}

\newcommand{\fq}{\mathbb{F}_{q}}
\newcommand{\fqD}{\fq[D_{2n}]}

\newcommand{\mC}{\mathcal{C}}

\newcommand{\mI}{\mathcal{I}}

\newcommand{\rk}{\textnormal{rk}}

\newcommand{\im}{\textnormal{im}}

\newcommand{\Mat}{\textnormal{Mat}}

\newcommand*{\myproofname}{Proof of the claim}

\makeatletter
\renewcommand*\env@matrix[1][*\c@MaxMatrixCols c]{%
  \hskip -\arraycolsep
  \let\@ifnextchar\new@ifnextchar
  \array{#1}}
\makeatother

\usepackage{hyperref}

\title{
\textbf{Dihedral Quantum Codes}}

\date{}                    

\usepackage{setspace}

 \allowdisplaybreaks
\usepackage{longtable}
\usepackage{multirow}

\usepackage{authblk}
\author[2]{Nadja Willenborg}
\author[1,3]{Martino Borello}
\author[2]{Anna-Lena Horlemann}
\author[2]{Habibul Islam}

\affil[1]{Universit\'e Paris 8, Laboratoire de G\'eom\'etrie, Analyse et Applications, LAGA, Universit\'e
Sorbonne Paris Nord, CNRS, UMR 7539, France}
\affil[2]{University of St.Gallen, Switzerland}
\affil[3]{INRIA, France}

\begin{document}

\maketitle
	
\thispagestyle{empty}
	
\begin{abstract}
We establish dihedral quantum codes of short block length, a class of CSS codes obtained by the lifted product construction. We present the code construction and give a formula for the code dimension, depending on the two classical codes that the CSS code is based on. We also give a lower bound on the code distance and construct an example of short dihedral quantum codes.
\end{abstract}

\section{Introduction} 
From \cite{calderbank1996good} it is well known that good quantum codes, correcting phase flip errors and bit flip errors, exist in the sense of having positive rate and linear minimum distance. However, in recent years, significant progress has been made in the theory of quantum LDPC (low density parity check) codes, i.e., codes with a sparse parity check matrix, achieving asymptotically good parameters. The constructions given in \cite{panteleev2021quantum, panteleev2022asymptotically} are based on lifted products over cyclic group algebras. \footnote{A group algebra $K[G]$ is formed from a group $G$ and a field $K$. Throughout this article we restrict to the finite field $\fq$. For more details see Section \ref{sec:groupalg}.} Lifted product is a construction which lifts matrices from a field $\fq$ to a larger algebra and has led to the discovery of families of quantum LDPC codes with both linearly growing distance and dimension. This article adds to this line of research by considering dihedral group algebras and shows how they naturally give rise to short nonabelian quantum MDPC (moderate density parity check) codes. \footnote{MDPC codes have parity check matrices with column weights in $O(\sqrt{N})$, where $N$ is the code length.}

Specifically, we will explore the construction of quantum codes using the dihedral group algebra $\fq[D_{2n}]$. Our focus on the dihedral group algebra allows us to use methods from \cite{martinez2015structure,vedenev2020relationship} to present a novel nonabelian quantum CSS code construction using dihedral groups. A CSS code is a quantum code built from two classical linear codes with additional orthogonality constraints on their parity check matrices. 

The paper is structured as follows: We start by introducing the necessary mathematical background. After defining lifted product codes, we start constructing dihedral lifted product codes from these in Section \ref{sec:dihedral}. In that section we also present the dimension formula and a distance bound for these codes. Afterwards, we use the dihedral group of order $180$ and are able to compute a concrete code example. In Section \ref{sec:final} we summarize our results and their potential future applications.
\\
\\
\textbf{Acknowledgments.}
 We would like to thank Markus Grassl, Virgile Guemard, Anthony Leverrier and Pavel Panteleev for useful comments. 
\section{Preliminaries}\label{sec:pre}

 In this paper, we adhere to classical notations for matrix spaces. Specifically, $\Mat_{m\times n}(\mathcal{H})$ denotes the space of $m \times n$-matrices with elements in an algebra $\mathcal{H}$ and we use $\Mat_{m}(\mathcal{H})$ when referring to square $m\times m$-matrices over $\mathcal{H}$. Furthermore, since we restrict ourselves to finite groups, our algebra will be denoted as $\fq[G]$, that is, the group algebra of $G$ over the finite field $\fq$. For any matrix $A$, we denote by $A^\top$ its transpose.

 The set of positive integers up to $n$ is denoted by $[n]$.

 \subsection{Classical Codes}
 A classical linear code $\mC$ with parameters $[n,k]_{q}$ is a $k$-dimensional vector space in $\fq^n$. The Hamming distance $d(v,v')$ between $v,v' \in \fq^n$ is the number of positions, where $v,v'$ differ. The parameter $$d(\mC)= \min \{d(v,v'): v \neq v', v,v' \in \mC\}$$
 is called the \emph{minimum (Hamming) distance} of $\mC$. A linear $[n,k]_{q}$ code $\mC$ with $d(\mC) =d$ is called an $[n,k,d]_{q}$ code. A linear $[n,k]_{q}$ code can be defined as the kernel of a matrix $H \in \fq^{(n-k) \times n}$ with $\rk H = n-k$, called a \emph{parity check matrix} of the code. The rows of $H$ are orthogonal to any vector in $\mC$.  The code defined by a parity check matrix $H$ is denoted by $\mC(H)$.

\subsection{Quantum CSS codes}
We consider the complex Hilbert space $\mathbb{C}^q$ of dimension $q$ and its $N$-fold tensor product $(\mathbb{C}^q)^{\otimes N}$, also known as \emph{$N$-qudit space}, where each factor $\mathbb{C}^{q}$ describes the state space of a single qudit. A quantum error correcting code of length $N$ and dimension $K$ is a $q^K$-dimensional subspace of $(\mathbb{C}^q)^{\otimes N}$; if it can correct up to $\lfloor (D-1)/2 \rfloor$ errors we denote it by $[[N,K,D]]_{q}$.
\emph{Calderbank-Shor-Steane} (CSS) codes form an important subclass of quantum error correcting codes \cite{calderbank1996good,steane1996simple}. A CSS code is defined by a pair of classical linear codes $\mC_X, \mC_Z \subseteq \fq^N$ and can be identified in the following way
\begin{equation} \label{eq:CSS_sum}
    Q(\mC_X,\mC_Z) := \mC_Z/ \mC_{X}^\perp \oplus \mC_X/ \mC_{Z}^\perp.
\end{equation}
The dimension of the CSS code $Q(\mC_X,\mC_Z)$ is $K=\dim \mC_X/ \mC_{Z}^\perp$, which can be reformulated as follows
\begin{equation}\label{eq:dimension}
    K = \dim \mC_{X} - \dim \mC_{Z}^\perp = \dim \ker H_{Z} - \rk H_{X} = N - \rk H_{Z}-  \rk H_{X}.
\end{equation}
Here, \(\mathcal{C}_X^{\perp}\) (resp.\ \(\mathcal{C}_Z^{\perp}\)) denotes the dual code of 
\(\mathcal{C}_X\) (resp.\ \(\mathcal{C}_Z\)), 
and \(\mathrm{rk}\,H_X\) (resp.\ \(\mathrm{rk}\,H_Z\)) is the rank of a parity check matrix 
for \(\mathcal{C}_X\) (resp.\ \(\mathcal{C}_Z\)).
On the other hand, given $K>0$, its minimum distance is given by $D= \min \{D_X, D_Z\},$ where
\begin{equation*}\label{eq:distance}
  D_Z:= \min_{c \in \mC_{Z}\setminus \mC_{X}^\perp}|c|, \quad D_X:= \min_{c \in \mC_{X}\setminus \mC_{Z}^\perp}|c|.
\end{equation*}
Even when $K=0$ the definition of the minimum distance remains well-defined. In fact, if $K=0$, it must be that either $\mC_{X}=\mC_{Z}^\perp$ or $\mC_{Z}=\mC_{X}^\perp.$ In the first case, the set 
$\mathcal{C}_X \setminus \mathcal{C}_Z^\perp$
is empty and we adopt the convention \(D_X = \infty\), so that \(D = D_Z\). In the second case, we set 
$D_Z = \infty$
and hence \(D = D_X\). In either scenario, provided the code is nontrivial (i.e., not the zero code), at least one of the sets 
$
\mathcal{C}_X \setminus \mathcal{C}_Z^\perp$ or $\mathcal{C}_Z \setminus \mathcal{C}_X^\perp$
is nonempty, ensuring that the minimum distance \(D\) is well defined.

Note that by considering two classical error correcting codes in the CSS construction, this definition enables the code to handle the two primary types of quantum errors: bit flip errors (quantified by $D_X$) and phase flip errors (quantified by $D_Z$).

To guarantee that the CSS code in \eqref{eq:CSS_sum} is well-defined, we need $\mC_{X}^\perp \subseteq \mC_Z$ (or equivalently $\mC_{Z}^\perp \subseteq \mC_X$). Let $H_{X}$ be a parity check matrix of $\mC_{X}$ and $H_{Z}$ be a parity check matrix of $\mC_{Z}$, then we can express this via the following orthogonality condition
\begin{equation} \label{eq:orthogonality}
    H_{X}H_{Z}^\top =0.
\end{equation}
Indeed, since the parity check matrix $H_X$ is the generator matrix of $\mC_{X}^\perp$ and $H_{Z}$ is a parity check matrix of $\mC_{Z}$, we have that the row space of $H_{X}$ is contained in $\mC_{Z}$, i.e., $\mC_{X}^{\perp} \subseteq \mC_{Z}.$

\subsection{Group codes} \label{sec:groupalg}

Let $G$ be a finite group with neutral element $e$  and $\fq$ be a field. The group algebra $\fq[G]$ over $\fq$ is the set 
$$\fq[G]:= \left\{ \sum_{g \in G}a_{g}g \mid a_{g} \in \fq\right\},$$
with the following operations for $a, b \in \fq[G]$, $a=\sum_{g \in G}a_{g}g, b=\sum_{g \in G}b_{g}g$, and $c\in \F_q$: 
$$a+b:= \sum_{g \in G}(a_{g}+b_{g})g,$$
$$c \cdot a := \sum_{g \in G}ca_{g}g,$$
$$a\cdot b:= \sum_{g \in G} \left( \sum_{\mu \nu = g}a_{\mu}b_{\nu} \right)g.$$

\begin{definition} \label{def:code_def}
  Let $G$ be a finite group with neutral element $e$. A \emph{left} group action of $G$ on $\mathbb{F}_q[G]$ is a function
  $\sigma: G \times \mathbb{F}_q[G] \to \mathbb{F}_q[G]$, satisfying
  \begin{itemize}
    \item $\sigma(e, x) = x$ for all $x \in \mathbb{F}_q[G]$, 
    \item $\sigma(gh, x) = \sigma\!\bigl(g,\,\sigma(h,x)\bigr)$ for all $g,h \in G$ and $x \in \mathbb{F}_q[G]$.
    \end{itemize}
    If in addition $\sigma$ is free, meaning that for every $g \in G \setminus\{e\}$ and every $x \in \mathbb{F}_q[G]$ we have $\sigma(g,x)\neq x$ then $\sigma$ is called a \emph{free left} group action.
    
 Let $\sigma: G \times \mathbb{F}_q[G] \to \mathbb{F}_q[G]$ be a free left group action. For a positive integer $\ell$, let 
  $\mathcal{C} \,\subseteq\, \bigl(\mathbb{F}_q[G]\bigr)^{\ell}$
  be an $\mathbb{F}_q$-submodule. 
  If $\mathcal{C}$ is \emph{invariant} under $\sigma$, \emph{i.e.},
  \[
    \forall\,g \in G \,\;\forall\,c \in \mathcal{C}
    \quad\colon\quad
    \sigma(g, c)\in \mathcal{C},
  \]
  where $\sigma$ acts componentwise on tuples in $\mathcal{C}$, then we call
  $\mathcal{C}$ a \emph{generalized quasi-group code of index $\ell$}. 
  If $\ell = 1$, then $\mathcal{C}$ is called a \emph{generalized group code}.
\end{definition}

We note that the above definitions can straight-forwardly be adopted to \emph{right} group actions.

\begin{remark}
    Note that the invariance of the quasi-group code $\mC$ under the group action $\sigma$ as defined implies invariance under specific types of group actions. If $\sigma$ acts in a manner similar to a left group action (i.e., $\sigma(g,x)=gx, \forall g \in G, \forall x \in \fq[G]$), then invariance under $\sigma$ implies invariance under the left group action $\lambda_{g} : x \mapsto gx$. On the other hand, if $\sigma$ behaves like the corresponding right group action (i.e., $\sigma(g,x)=g^{-1}x$), then invariance under $\sigma$ corresponds to invariance under the right group action $\rho_{g}: x \mapsto g^{-1}x)$. Consequently, every quasi-group code \( \mC \) can be regarded either as a left module (i.e., invariant under $\lambda_{g}$) or right module (i.e., invariant under $\rho_{g}$), depending on the specific behaviour of \( \sigma \) and which group action—\( \lambda_g \) or \( \rho_g \)—is more relevant for the context considered.
\end{remark}

To represent group codes as codes over $\F_q$ and use them in the CSS construction, we will use the following representations:
\begin{definition} \label{def:matrix_rep}
    Let $a \in \fq[G]$. The \emph{right} (resp.\ \emph{left}) \emph{regular matrix representation} with respect to a fixed basis of $\fq^{|G|}$
    is defined as the $|G| \times |G|$-matrix of the linear operator $\rho_{a} : \fq^{|G|} \rightarrow \fq^{|G|}, x \mapsto xa$ (resp.\ $\lambda_{a} : \fq^{|G|} \rightarrow \fq^{|G|}, x \mapsto ax$). We denote the right regular matrix representation of $a$ by $R(a)$ and its left regular matrix representation by $L(a).$ Clearly, when $\fq[G]$ is commutative we do not need to distinguish between left and right regular representations. In this case we simply denote the corresponding matrix over $\fq$ by $\mathbf{a}$, and for $A \in \Mat_{m}(\fq[G])$ we denote its corresponding matrix over $\fq$ by $\mathbf{A}$.
\end{definition}

The following will be useful in our code construction:
\begin{proposition} \label{prop:comm}
    For any $a,b \in \fq[G]$
    $$L(a)^\top R(b)= R(b)L(a)^\top.$$
    \begin{proof}
     Since multiplication in $\fq[G]$ is associative we have for any $x \in \fq[G]$ that
     $a(xb)=(ax)b$ holds. This shows for any $x \in \fq^{|G|}$ that applying $L(a)^\top$ after $R(b)$ on $x$ gives the same result as first applying $L(a)^\top$ on $x$ and then $R(b).$
    \end{proof}
\end{proposition}

It is well known (see for example \cite[Chapter 16]{huffman2021concise}) that  if $\textnormal{char} (\fq) \nmid |G| $, all group codes over $\fq[G]$ are \emph{principal}, i.e., 
that there exists $c \in \fq[G]$ with $\mC = \fq[G]c.$

\subsection{Lifted product construction}

The lifted product was introduced in \cite{panteleev2021quantum} and formalizes many known constructions of quantum codes. The idea is to lift the elements in matrices over $\fq$ up to some ring $R$ that is also a finite dimensional $\fq$-algebra. To define these codes we need the Kronecker product over $\fq[G]$ and the conjugate transpose of a matrix $H \in \Mat_{m \times n}(\fq[G])$. We start by defining these concepts in the context of group algebras. 

\begin{definition} \label{def:conjugate}
Let $a\in \fq[G]$ such that $\displaystyle a= \sum_{g\in G}a_{g}g$. Then its \emph{reciprocal} $a^{*}$ is defined as
\begin{equation}\label{eq:involution}
    a^{*} := \sum_{g\in G}a_{g^{-1}}g.
\end{equation} 
    If $H=(h_{i,j})_{1 \le i\le m, 1\le j \le n}$ is a matrix over $\fq [G]$ we define its \emph{conjugate transpose} as $H^{*}:= (h_{j,i}^{*})_{1 \le j \le n, 1 \le i \le m}$, where $h_{j,i}^{*}$ is the reciprocal of $h_{i,j} \in \fq[G]$.
\end{definition}
\begin{remark} \label{rem:conjugate} 
\label{rem:reciprocal}
Note that the Hamming weight is invariant under taking the reciprocal.
Moreover, to see that \((a + b)^* = a^* + b^*\) for any \(a,b \in \mathbb{F}_q[G]\),
observe that the map \(g \mapsto g^{-1}\) is a bijection on \(G\). Hence, we can
reindex the sum in \((a + b)^*\) by replacing each summation index \(g\) with \(g^{-1}\),
which shows
\[
  (a + b)^*
  \;=\;
  \sum_{g \in G} (a_g + b_g)\,g^{-1}
  \;=\;
  \sum_{g \in G} a_g\,g^{-1}
  \;+\;
  \sum_{g \in G} b_g\,g^{-1}
  \;=\;
  a^* + b^*.
\]
\end{remark}

\begin{lemma} Let $a \in \mathbb{F}_{q}[G]$. Then
$$R(a^{*})=R(a)^\top \text{ and  }L(a^{*})=L(a)^\top.$$
\end{lemma}
\begin{proof}
In fact, since $g_{i}g=g_{j}$ implies $g= g_{i}^{-1}g_{j}$, the coefficient of $g_{j}$ in the product $g_{i}a$ is $a_{g_{i}^{-1}g_{j}}$. On the other hand, the coefficient of $g_{i}$ in $g_{j}a^{*}$ is $a_{(g_{j}^{-1}g_{i})^{-1}}$. Hence the two coefficients are equal, which shows that the $(j,i)$-th entry of $L(a)$ is equal to the $(i,j)$-th entry of $L(a^{*})$. 
\end{proof}

The lifted product construction is based on the well known Kronecker product:
\begin{definition}\label{def:Kronecker}
Let $A \in \textnormal{Mat}_{m_{A}\times n_{A}}(\fq[G])$ and $B \in \textnormal{Mat}_{m_{B}\times n_{B}}(\fq[G])$. Then the \emph{Kronecker product} of $A$ and $B$ is the $m_{A}m_{B} \times n_{A}n_{B}$ block matrix $A \otimes B$ given by
$$A \otimes B = \begin{pmatrix}
a_{11}B & a_{12}B & \cdots & a_{1n}B \\
a_{21}B & a_{22}B & \cdots & a_{2n}B \\
\vdots  & \vdots  & \ddots & \vdots  \\
a_{m1}B & a_{m2}B & \cdots & a_{mn}B
\end{pmatrix}.$$
\end{definition}

\begin{proposition} \label{prop:well-def}
Let $A \in \textnormal{Mat}_{m_{A}\times n_{A}}(\fq[G]), B \in \textnormal{Mat}_{m_{B}\times n_{B}}(\fq[G])$ and define
\begin{align*}
    {}^{\natural}A&:= [L(a_{ij})^\top]_{1 \le i \le m_{A}, 1 \le j \le n_{A}} \in \Mat_{m_{A} \times n_{A}}(\fq^{|G| \times |G|}),\\ \quad B^{\natural}&:= [R(b_{ij})]_{1 \le i \le m_{B}, 1 \le j \le n_{B}} \in \Mat_{m_{B} \times n_{B}}(\fq^{|G| \times |G|}).
\end{align*}

Moreover we form block matrices\footnote{Here the notation $[A,B]$ denotes the $m \times (n_{A}+n_{B})$ block matrix by placing $B$ to the right of $A$. 
}
\begin{equation} \label{eq:parity_fq}
    H_{X}^{\natural} := \left[  {}^{\natural}A \otimes I_{m_B}, -I_{m_A} \otimes B^{\natural} \right], \quad H_Z^{\natural} := \left[ I_{n_A} \otimes {B^{\natural}}^\top, { {}^{\natural}A}^\top \otimes I_{n_B} \right].
\end{equation}
Then we have
$$ H_X^{\natural} {H_Z^\natural}^\top = 0 .$$ 
 \end{proposition}
 \begin{proof}
We obtain
\begin{align*}
   H_{X}^{\natural}{H_{Z}^{\natural}}^\top&= (A^{\natural} \otimes I_{m_{B}})(I_{n_{A}}\otimes B^{\natural})+(-I_{m_{A}}\otimes B^{\natural})(A^{\natural} \otimes I_{n_{B}})\\
   &= 
\scalebox{0.75}{$\left(
\begin{array}{c|c|c}
\begin{matrix}
    a_{11}^{\natural} & \cdots &  \mathbf{0} \\
    \vdots  & \ddots & \vdots \\
     \mathbf{0} &   \cdots &  a_{11}^{\natural}
\end{matrix} &\cdots & \begin{matrix}
    a_{1n_{A}}^{\natural}& \cdots &  \mathbf{0} \\
    \vdots  & \ddots & \vdots \\
    \mathbf{0} &   \cdots &  a_{1n_{A}}^{\natural}
\end{matrix} \\
\hline
\vdots  & \ddots & \vdots\\
\hline
\begin{matrix}
    a_{m_{A}1}^{\natural}& \cdots &  \mathbf{0} \\
    \vdots  & \ddots & \vdots \\
     \mathbf{0} &   \cdots &  a_{m_{A}1}^{\natural}
\end{matrix} &\cdots & \begin{matrix}
    a_{m_{A}n_{A}}^{\natural}& \cdots &  \mathbf{0} \\
    \vdots  & \ddots & \vdots \\
   \mathbf{0}&   \cdots &  a_{m_{A}n_{A}}^{\natural}
\end{matrix} 
\end{array}
\right)
\left(
\begin{array}{c|c|c}
\begin{matrix}
    b_{11}^{\natural} & \cdots &  b_{1n_{B}}^{\natural} \\
    \vdots  & \ddots & \vdots \\
     b_{m_{B}1}^{\natural} &   \cdots &  b_{m_{B}n_{B}}^{\natural}
\end{matrix} &\cdots & \begin{matrix}
\mathbf{0}
\end{matrix} \\
\hline
\vdots  & \ddots & \vdots\\
\hline
\begin{matrix}
  \mathbf{0}
\end{matrix} &\cdots & \begin{matrix}
     b_{11}^{\natural} & \cdots &  b_{1n_{B}}^{\natural} \\
    \vdots  & \ddots & \vdots \\
     b_{m_{B}1}^{\natural} &   \cdots &  b_{m_{B}n_{B}}^{\natural}
\end{matrix} 
\end{array}
\right)$}
\\
\\
&- \scalebox{0.75}{$\left(
\begin{array}{c|c|c}
\begin{matrix}
    b_{11}^{\natural} & \cdots &  b_{1n_{B}}^{\natural} \\
    \vdots  & \ddots & \vdots \\
     b_{m_{B}1}^{\natural} &   \cdots &  b_{m_{B}n_{B}}^{\natural}
\end{matrix} &\cdots & \begin{matrix}
\mathbf{0}
\end{matrix} \\
\hline
\vdots  & \ddots & \vdots\\
\hline
\begin{matrix}
  \mathbf{0}
\end{matrix} &\cdots & \begin{matrix}
     b_{11}^{\natural} & \cdots &  b_{1n_{B}}^{\natural} \\
    \vdots  & \ddots & \vdots \\
     b_{m_{B}1}^{\natural} &   \cdots &  b_{m_{B}n_{B}}^{\natural}
\end{matrix} 
\end{array}
\right)
\left(
\begin{array}{c|c|c}
\begin{matrix}
    a_{11}^{\natural} & \cdots &  \mathbf{0} \\
    \vdots  & \ddots & \vdots \\
     \mathbf{0} &   \cdots &  a_{11}^{\natural}
\end{matrix} &\cdots & \begin{matrix}
    a_{1n_{A}}^{\natural}& \cdots &  \mathbf{0} \\
    \vdots  & \ddots & \vdots \\
    \mathbf{0} &   \cdots &  a_{1n_{A}}^{\natural}
\end{matrix} \\
\hline
\vdots  & \ddots & \vdots\\
\hline
\begin{matrix}
    a_{m_{A}1}^{\natural}& \cdots &  \mathbf{0} \\
    \vdots  & \ddots & \vdots \\
     \mathbf{0} &   \cdots &  a_{m_{A}1}^{\natural}
\end{matrix} &\cdots & \begin{matrix}
    a_{m_{A}n_{A}}^{\natural}& \cdots &  \mathbf{0} \\
    \vdots  & \ddots & \vdots \\
   \mathbf{0}&   \cdots &  a_{m_{A}n_{A}}^{\natural}
\end{matrix} 
\end{array}
\right)$}
 \\
 \\
&= \scalebox{0.75}{$\left(
\begin{array}{c|c|c|c}
\begin{matrix}
    a_{11}^{\natural}b_{11}^{\natural}& \cdots &  a_{11}^{\natural}b_{1n_{B}}^{\natural} \\
    \vdots  & \ddots & \vdots \\
     a_{11}^{\natural}b_{m_{B}1}^{\natural} &   \cdots &  a_{11}^{\natural}b_{m_{B}n_{B}}^{\natural}
\end{matrix} & \begin{matrix}
    a_{12}^{\natural}b_{11}^{\natural}& \cdots &  a_{12}^{\natural}b_{1n_{B}}^{\natural} \\
    \vdots  & \ddots & \vdots \\
     a_{12}^{\natural}b_{m_{B}1}^{\natural} &   \cdots &  a_{12}^{\natural}b_{m_{B}n_{B}}^{\natural}
\end{matrix} &\cdots & \begin{matrix}
    a_{1n_{A}}^{\natural}b_{11}^{\natural}& \cdots &  a_{1n_{A}}^{\natural}b_{1n_{B}}^{\natural} \\
    \vdots  & \ddots & \vdots \\
     a_{1n_{A}}^{\natural}b_{m_{B}1}^{\natural} &   \cdots &  a_{1n_{A}}^{\natural}b_{m_{B}n_{B}}^{\natural}
\end{matrix} \\
\hline
\begin{matrix}
    a_{21}^{\natural}b_{11}^{\natural}& \cdots &  a_{21}^{\natural}b_{1n_{B}}^{\natural} \\
    \vdots  & \ddots & \vdots \\
     a_{21}^{\natural}b_{m_{B}1}^{\natural} &   \cdots &  a_{21}^{\natural}b_{m_{B}n_{B}}^{\natural}
\end{matrix}  & \vdots & \ddots & \vdots \\
\hline
\vdots & \vdots & \ddots & \vdots\\
\hline
\begin{matrix}
    a_{m_{A}1}^{\natural}b_{11}^{\natural}& \cdots & a_{m_{A}1}^{\natural}b_{1n_{B}}^{\natural} \\
    \vdots  & \ddots & \vdots \\
     a_{m_{A}1}^{\natural}b_{m_{B}1}^{\natural} &   \cdots &  a_{m_{A}1}^{\natural}b_{m_{B}n_{B}}^{\natural}
\end{matrix} &\cdots & \cdots & \begin{matrix}
    a_{m_{A}n_{A}}^{\natural}b_{11}^{\natural}& \cdots &  a_{m_{A}n_{A}}^{\natural}b_{1n_{B}}^{\natural} \\
    \vdots  & \ddots & \vdots \\
     a_{m_{A}n_{A}}^{\natural}b_{m_{B}1}^{\natural} &   \cdots &  a_{m_{A}n_{A}}^{\natural}b_{m_{B}n_{B}}^{\natural}
\end{matrix} 
\end{array}
\right)$}
\\
\\
&- \scalebox{0.75}{$\left(
\begin{array}{c|c|c|c}
\begin{matrix}
    b_{11}^{\natural} a_{11}^{\natural} & \cdots & b_{1n_B}^{\natural} a_{11}^{\natural} \\
    \vdots & \ddots & \vdots \\
    b_{m_B1}^{\natural} a_{11}^{\natural} & \cdots & b_{m_Bn_B}^{\natural} a_{11}^{\natural}
\end{matrix} &
\begin{matrix}
    b_{11}^{\natural} a_{12}^{\natural} & \cdots & b_{1n_B}^{\natural} a_{12}^{\natural} \\
    \vdots & \ddots & \vdots \\
    b_{m_B1}^{\natural} a_{12}^{\natural} & \cdots & b_{m_Bn_B}^{\natural} a_{12}^{\natural}
\end{matrix} & \cdots & 
\begin{matrix}
    b_{11}^{\natural} a_{1n_A}^{\natural} & \cdots & b_{1n_B}^{\natural} a_{1n_A}^{\natural} \\
    \vdots & \ddots & \vdots \\
    b_{m_B1}^{\natural} a_{1n_A}^{\natural} & \cdots & b_{m_Bn_B}^{\natural} a_{1n_A}^{\natural}
\end{matrix} \\
\hline
\begin{matrix}
    b_{11}^{\natural} a_{21}^{\natural} & \cdots & b_{1n_B}^{\natural} a_{21}^{\natural} \\
    \vdots & \ddots & \vdots \\
    b_{m_B1}^{\natural} a_{21}^{\natural} & \cdots & b_{m_Bn_B}^{\natural} a_{21}^{\natural}
\end{matrix} & \vdots & \ddots & \vdots \\
\hline
\vdots & \vdots & \ddots & \vdots \\
\hline
\begin{matrix}
   b_{11}^{\natural} a_{m_A1}^{\natural} & \cdots & b_{1n_B}^{\natural} a_{m_A1}^{\natural} \\
    \vdots & \ddots & \vdots \\
   b_{m_B1}^{\natural} a_{m_A1}^{\natural} & \cdots & b_{m_Bn_B}^{\natural} a_{m_A1}^{\natural}
\end{matrix} & \cdots & \cdots & 
\begin{matrix}
    b_{11}^{\natural} a_{m_An_A}^{\natural} & \cdots & b_{1n_B}^{\natural} a_{m_An_A}^{\natural} \\
    \vdots & \ddots & \vdots \\
    b_{m_B1}^{\natural} a_{m_An_A}^{\natural} & \cdots & b_{m_Bn_B}^{\natural} a_{m_An_A}^{\natural}
\end{matrix} 
\end{array}
\right)$}\\
&= 0,
\end{align*}
where the last equality follows from Proposition \ref{prop:comm}, using that ${}^{\natural}a_{ij}=L(a_{i,j})^\top, b_{ij}^{\natural}=R(b_{i,j})$. 
      \end{proof}

\begin{definition}\label{def:LP} (Lifted Product Construction, see \cite{panteleev2021quantum})
       Let $A \in \textnormal{Mat}_{m_{A}\times n_{A}}(\fq[G])$ and $B \in \textnormal{Mat}_{m_{B}\times n_{B}}(\fq[G])$, and define matrices over $\fq$ by 
   \begin{equation*}
H_X^\natural = \Bigl[ {}^\natural A \otimes I_{m_B},\; -I_{m_A}\otimes B^\natural \Bigr], \quad
H_Z^\natural = \Bigl[ I_{n_A}\otimes (B^\natural)^\top,\; ({}^\natural A)^\top\otimes I_{n_B} \Bigr],
\end{equation*}
where ${}^\natural A := \Bigl[ L(a_{ij})^\top \Bigr]_{1\le i\le m_A,\, 1\le j\le n_A} \quad\text{and}\quad B^\natural := \Bigl[ R(b_{ij}) \Bigr]_{1\le i\le m_B,\, 1\le j\le n_B}$ are the matrices where each entry $a_{ij}$ of $A$ (respectively $b_{ij}$ of $B$) is replaced by its left, respectively right regular matrix representation. We then define the lifted product code $LP(A,B)$ as the quantum CSS code with parity-check matrices \(H_X^\natural\) and \(H_Z^\natural\). Proposition \ref{prop:well-def} guarantees that $LP(A,B)$ is a well-defined quantum CSS code. 

\end{definition}

\begin{remark}
    Definition \ref{def:LP} is a reformulation of the lifted product construction introduced in \cite{panteleev2021quantum,panteleev2022asymptotically}. In these works, the authors introduce a method based on chain complexes to produce new quantum LDPC codes by taking a tensor product over a finite dimensional algebra $R$. Specifically, our requirement that $A$ be a right $R$-module and $B$ be a left $R$-module commuting with the free group action, reproduces the lifted product codes from \cite{panteleev2021quantum,panteleev2022asymptotically} when $R$ is the group algebra $\fq[G]$. Thus Definition \ref{def:LP} can be seen as a natural, more concrete restatement of their general framework in the context of free group actions.
    \end{remark}
  One interesting fact about lifted product codes as defined above is that they are moderate density parity check (MDPC) codes\footnote{These codes are particularly interesting in the area of code based cryptography.} (i.e., codes who have a parity check matrix whose rows have Hamming weights in $O(\sqrt{N})$, see e.g. \cite{bariffi2022moderate}), if $m_B$ and $m_A$ are in the same order of magnitude as $n_A$ and $n_B$, respectively. We will prove the statement for the case $m_B=n_A, m_A=n_B$, but the analog holds for $m_B\in O(n_A), m_A\in O(n_B)$. 
     
\begin{proposition}
    Let  $A \in \textnormal{Mat}_{n_{B}\times n_{A}}(\fq[G])$ and $B \in \textnormal{Mat}_{n_{A}\times n_{B}}(\fq[G])$. Then an $LP(A,B)$ code with parity check matrices $H_{X}, H_{Z}$, as defined in \eqref{eq:parity_fq}, is a moderate density parity check (MDPC) code.
    \end{proposition}
    \begin{proof}
     Again we create the matrices $$A^{\natural} \in \Mat_{n_{B} \times n_{A}}(\fq^{|G| \times |G|}), \quad B^{\natural} \in \Mat_{n_{A} \times n_{B}}(\fq^{|G| \times |G|}),$$
by replacing the elements $a_{ij}$ of $A$, (respectively $b_{ij}$ of $B$) by $L(a_{ij})^\top$, (respectively $R(b_{ij}))$ and the analogs of $H_{X}$ and $H_{Z}$ in the following way
$$H_{X}^{\natural} = \left[ A^{\natural} \otimes I_{n_A}, -I_{n_B} \otimes B^{\natural} \right], \quad H_Z^{\natural} = \left[ I_{n_A} \otimes {B^{\natural}}^\top, {A^{\natural}}^\top \otimes I_{n_B} \right].
$$  Let $\omega_{X}, \omega_{Z}$ be the maximal row weights of $H_{X}^\natural, H_{Z}^\natural$ 
       and let $N$ be the length of the $LP(A,B)$ code, i.e., $N= (n_{A}^2 + n_{B}^2)|G|$. To show that the $LP(A,B)$ code is MDPC we use Definition 2.2 from \cite{bariffi2022moderate} and show that $\omega_{X}, \omega_{Z} \in O(\sqrt{N})$, as $N \to \infty.$ \footnote{Note that we fix the order of the group and let either the row length $n_{A}$ of $A$ or the row length $n_{B}$ of $B$ go to infinity.} We easily see that the row weights of the parity check matrices are upper bounded by $(n_{A}+n_{B})|G|.$ Since
       $$\limsup_{n_{\bullet} \to \infty}\frac{|G|^2(n_{A}^2 +2n_{A}n_{B}+n_{B}^2)}{|G|(n_{A}^{2}+n_{B}^{2})}< \infty \quad \text{for} \quad \bullet \in \{A,B\},$$
       we have $\omega_{X}^{2}, \omega_{Z}^{2} \in O(N)$ and the statement follows.
    \end{proof}

\section{Dihedral lifted product codes}\label{sec:dihedral}
We now present our main results: the parameters of dihedral lifted product codes. 
Throughout this section 
we assume $\textnormal{char}(\fq)\nmid |G|$ and $n \ge 2$.

The \emph{cyclic group} of order $n$, containing $1, \alpha, \alpha^{2}, \dots, \alpha^{n-1}$, denoted by $C_{n}$, is defined as $C_{n}:= \langle \alpha \rangle$, where $\alpha^n =1$ and $\alpha^m \neq 1$ for $0<m <n$.
The \emph{dihedral group} of order $2n$, containing $1, \alpha, \alpha^2, \dots, \alpha^{n-1}, \beta, \alpha \beta, \alpha^{2}\beta, \dots, \alpha^{n-1}\beta$, denoted by $D_{2n}$, is defined as $D_{2n}:= \langle \alpha, \beta \rangle$, where $\alpha^{n} = \beta^2 = 1$ and $\beta \alpha = \alpha^{n-1}\beta, \alpha^{m} \neq 1$, for $0<m <n$.

Let $f(x)=x^n+a_{n-1}x^{n-1}+\cdots+a_0\in \F_q[x],$ with \(a_0\neq 0\) and assume without loss of generality that \(f(x)\) is monic. We define the 
\emph{normalized reciprocal polynomial} of \(f(x)\) as
\[
f^*(x):=\frac{1}{a_0}\,x^n\,f\left(\frac{1}{x}\right).
\]
and say that \(f\) is \emph{self-reciprocal} if 
$f^*=f$. 
Now, we factorize \begin{equation*} \label{eq:fac}
    x^n -1 = \prod_{i=1}^r f_{i} \prod_{i=r+1}^{r+s}f_{i}^{*}f_{i},
\end{equation*} where $r$ is the number of self-reciprocal factors and $2s$ the number of non-self-reciprocal
factors. Let
$$\theta (n) := \begin{cases}
  1& \text{if $n$ is odd} \\
2 & \text{if $n$ is even} \\
\end{cases}
$$
and let $\fq \subseteq F_{i}$ be extension fields of $\fq$ such that $[F_{i}: \fq] =\deg f_{i}/2$ if $\theta(n)+1 \le i \le r$ and
$[F_{i}: \fq] = \deg f_{i}$ in all other cases. In \cite{martinez2015structure,vedenev2020relationship} explicit decompositions of $\fqD$ were obtained. For our construction we use a more generic framework and decompose $\fqD$ as
\begin{equation} \label{eq:dec_generic}
    \fqD \cong \bigoplus_{i=1}^{r+s}R_{i},
\end{equation}
where 
$$R_{i}= \begin{cases}\fq[C_{2}]
 & \text{if $1 \le i \le \theta(n)$} \\
 \Mat_{2}(F_{i})
 & \text{if $\theta(n)+1 \le i \le r+s$} \\ 
\end{cases}. $$
It is known (see \cite{vedenev2019codes}) that a code $\mC \subseteq \fqD$ in the direct sum \eqref{eq:dec_generic} of algebras is the direct sum of left ideals in the terms. In particular the code admits a canonical decomposition
\begin{equation*} \label{eq:code_dec}
        \mC \cong \bigoplus_{i=1}^{r+s}C_{i},
    \end{equation*}
    with each component $C_{i}$ being a submodule of the simple algebra $R_{i}.$ More precisely, for each index $i$ one of the following cases occurs:
    \begin{itemize}
\item \(C_i = R_i\),
    \item \(C_i = I_i\), where $I_i \subsetneq R_i$ is a proper ideal of $R_{i}$, or
    \item \(C_i = 0\).
    \end{itemize}
Similar to \cite{vedenev2020relationship} we introduce the disjoint index sets
\[
J_1 = \{\, i \in [r+s] : C_i = A_i \,\} \quad \text{and} \quad J_2 = \{\, i \in [r+s] : C_i = I_i \,\},
\]
which we refer to as the \emph{corresponding sets} of the code. (See \cite[Theorem~2]{vedenev2020relationship} for a detailed account of this decomposition.)

\subsection{Dimension formula}
We start with some preliminary results that facilitate the exploration of our main result of this section in Theorem \ref{thm:dimension}. To derive a dimension formula we analyze the matrices $A,B$ defined over $\fq$, stemming from the decomposition outlined in Equation \eqref{eq:dec_generic}. In the following, we only consider the case of $A,B \in \Mat_{m}(\fq)$ to obtain good distance properties as illustrated in the subsequent sections. 

Note that $\fq$ can also be seen as the group algebra over the trivial group. 
\begin{lemma} \label{lem:rk_H}
    Let $A, B \in \Mat_{m}(\fq), k_{A}:= \dim \ker A, k_{B}:= \dim \ker B$ and let $H_{X}, H_{Z}$ be defined as in \eqref{eq:parity_fq}. 
    Then we have
    \begin{align*}
        \rk H_{X} =m^{2}-k_{A}k_{B}, \quad \rk H_{Z}=m^{2}-k_{A} k_{B}.
    \end{align*}
\end{lemma}
    \begin{proof}
    Let $A,B \in \Mat_{m}(\fq)$ and denote $r_{A}:=\rk(A), r_{B}:=\rk(B), k_{A}=m-r_{A}, k_{B}=m-r_{B}.$ It is well-known that for any two matrices, the Kronecker product satisfies $\rk(A \otimes I_{m})= m \cdot r_{A},$ respectively $\rk(B \otimes I_{m})= m \cdot r_{B}.$ Moreover, we have the identifications $\im(A \otimes I_{m})= \im (A) \otimes \fq^m$ and $\im (I_{m} \otimes B)= \fq^m \otimes \im (B).$ Since the column space of the block matrix $H_{X}=[A\otimes I_{m},-I_{m} \otimes B]$ is the sum of the column spaces of its blocks, we obtain $\rk(H_{X})=\dim(\im(A \otimes I_{m})+\im (I_{m} \otimes B)).$ Now, note that $$\im(A \otimes I_{m}) \cap \im(I_{m} \otimes B)=\im(A)\otimes\im(B).$$
    Hence $$\dim(\im(A) \otimes\im(B))=r_{A}r_{B}.$$ Thus, by the standard formula for the dimension of a sum of subspaces, we have $$\rk(H_{X})=mr_{A}+mr_{B}-r_{A}r_{B}.$$ Substituting $r_{A}=m-k_{A},r_{B}=m-k_{B},$ we compute $$mr_{A}+mr_{B}-r_{A}r_{B}=m(m-k_{A})+m(m-k_{B})-(m-k_{A}(m-k_{B}).$$ Expanding
    $$m(m-k_{A})+m(m-k_{B})=2m^{2}-m(k_{A}+k_{B})$$ and $$(m-k_{A})(m-k_{B})=m^{2}-m(k_{A}+k_{B})+k_{A}k_{B}.$$Hence $$\rk(H_{X})=(2m^{2}-m(k_{A}+k_{B}))-(m^{2}-m(k_{A}+k_{B})+k_{A}k_{B})=m^{2}-k_{A}k_{B}.$$ An entirely analogous argument shows that $\rk(H_{Z})=m^{2}-k_{A}k_{B}.$ This concludes the proof.
   
           \end{proof}

\begin{proposition} \label{prop:prop_dec}
    Let $A, B \in \Mat_{m}(\fq)$ and $k_{A}:= \dim \ker A, k_{B}:= \dim \ker B$. Then
    $$\dim LP(A,B) = 2k_{A}k_{B}.$$
\end{proposition}
    \begin{proof}
        From Lemma \ref{lem:rk_H} we have $\rk H_X = m^2 -k_{A}k_{B}$ and $\rk H_Z = m^2 -k_{A}k_{B}$. Hence using formula \eqref{eq:dimension} for the quantum dimension we obtain
        $$K= 2m^{2}-\rk H_X-\rk H_Z= 2k_{A}k_{B}. $$
    \end{proof}
In the following we present a dimension formula for dihedral lifted product codes. While dimension formulas for certain group algebras exist in the literature - especially in the abelian (cyclic) case, see \cite{panteleev2021quantum} - the explicit expression we derive here for nonabelian dihedral group algebras does not appear to have been documented before.
\begin{theorem} \label{thm:dimension}
Let $x^{n}-1= \left(\prod_{i=1}^{r}f_{i}(x) \right)\left(\prod_{i=r+1}^{r+s}f_{i}(x)f_{i}^{*}(x) \right)$ be the factorization of $x^{n}-1$ into irreducible factors. Let $e_{A}, e_B \in \fqD$ be two generating idempotents, such that $$e_{A}\fq[D_{2n}]=\bigoplus_{i=1}^{r+s}C_{i}^{A}, \quad e_{B}\fq[D_{2n}]=\bigoplus_{i=1}^{r+s}C_{i}^{B}$$ are the code decompositions, where
$$C_{i}^{\bullet}=\begin{cases}
    R_{i}& \text{if $i \in J_{1}^{\bullet}$} \\
I_{i} & \text{if $i \in J_{2}^{\bullet}$} \\
0 & \text{if $i \notin J_{1}^{\bullet} \cup J_{2}^{\bullet}$} \\
\end{cases},$$
for $\bullet \in \{A,B\}.$
Let $\mathbf{a}= (a_1 e_{A}, \dots, a_m e_{A}), \mathbf{b}= (b_1 e_{B}, \dots, b_m e_{B})$ and $a_{1}, \dots, a_{m}, b_{1}, \dots, b_{m} \in \fqD$ such that $a_{i}\fq[D_{2n}]+e_A\fq[D_{2n}]=\fq[D_{2n}],$ respectively $b_{i}\fq[D_{2n}]+e_B\fq[D_{2n}]=\fq[D_{2n}].$ \footnote{A straightforward way is to check a random $a_{i}$ until $a_{i}\fq[D_{2n}]+e_A\fq[D_{2n}]=\fq[D_{2n}]$ or its equivalent $x \notin e_{A}\fq[D_{2n}]$ is satisfied and repeat this process to get $a_{i}'s$ belonging to the correct coset for each row of $A$. An analogous method applies to $B \in \Mat_{m}(\mathbb{F}_{11}[D_{2n}])$ with its corresponding idempotent $e_{B}$. That suitable vectors $a,b$ are produced in a finite number of steps is guaranteed by standard ring-theoretic arguments (see \cite{rowen1988ring,lam1991first}).} Let $A, B \in \Mat_{m}(\fq[D_{2n}])$ such that all rows lie in $\mathbf{a}\fq[D_{2n}]$, respectively $\mathbf{b}\fq[D_{2n}]$. Then
  \begin{align*}
    \dim LP(A,B)=&\sum_{j=1}^{r}\deg f_{j}  (\mathbf{1}_{\mathcal{J}_{1}}(j)\omega_1 + \mathbf{1}_{\mathcal{J}_{2}}(j)\omega_2
    +\mathbf{1}_{\mathcal{J}_{3}}(j)\omega_3+\mathbf{1}_{\mathcal{J}_{4}}(j)\omega_4
    +\mathbf{1}_{\mathcal{J}_{5}}(j)\omega_5 +\mathbf{1}_{\mathcal{J}_{6}}(j)\omega_6) \\
    &+ \sum_{j=r+1}^{r+s}2\deg f_{j}( \mathbf{1}_{\mathcal{J}_{1}}(j)\omega_1 + \mathbf{1}_{\mathcal{J}_{2}}(j)\omega_2
    +\mathbf{1}_{\mathcal{J}_{3}}(j)\omega_3+\mathbf{1}_{\mathcal{J}_{4}}(j)\omega_4
    +\mathbf{1}_{\mathcal{J}_{5}}(j)\omega_5 +\mathbf{1}_{\mathcal{J}_{6}}(j)\omega_6),
\end{align*}
where
\begin{align*}
    \mathcal{J}_{1}&:= J_{2}^{A} \cap J_{2}^{B}\\
    \mathcal{J}_{2}&:= (J_{2}^{A} \cap J_{1}^{B}) \cup (J_{2}^{B} \cap J_{1}^{A})\\
    \mathcal{J}_{3} &:= J_{1}^{A} \cap J_{1}^{B}\\
    \mathcal{J}_{4} &:= [r+s]\setminus (J_{1}^{A} \cup J_{2}^{A}) \cap [r+s] \setminus (J_{1}^{B} \cup J_{2}^{B})\\
      \mathcal{J}_{5} &:= \left( J_{1}^{A} \cap [r+s]\setminus (J_{1}^{B} \cup J_{2}^{B})\right) \cup\left( J_{1}^{B} \cap [r+s] \setminus (J_{1}^{A} \cup J_{2}^{A}) \right)\\
        \mathcal{J}_{6} &:= ([r+s]\setminus (J_{1}^{A} \cup J_{2}^{A}) \cap (J_{2}^{B}) \cup ([r+s]\setminus (J_{1}^{B} \cup J_{2}^{B}) \cap (J_{2}^{A})\\
            \omega_1 &:=(2m-1)^{2}\\
      \omega_2 &:=2(m- 1)(2m-1)\\
      \omega_3 &:=2^{2}(m-1)^{2}\\
    \omega_4 &:=2^{2}m^{2}\\
    \omega_5 &:=2^{2}m(m-1)\\
    \omega_6 &:=2m(2m-1).
\end{align*}
    \begin{proof}
    Let $$e_{A}\fqD = \bigoplus_{i=1}^{r+s}C_{i}^A, \quad e_{B}\fqD = \bigoplus_{i=1}^{r+s}C_{i}^B$$
    be the decompositions of the codes $e_A\fqD, e_B\fqD$. Using the decomposition of $\fq[D_{2n}]$ any matrix $A \in \Mat_{m}(\fqD)$ can be uniquely represented by the collection of matrices $(A_{i})_{i \in [r+s]}$, where $A_{i}$ is the corresponding matrix over $R_{i}$, see \cite{vedenev2019codes} for more details. This gives 
        \begin{align} \label{eq:dec_formula}
            \dim LP(A,B) = \sum_{i=1}^{\theta (n)}\dim_{F_{i}}LP(A_i, B_i) &+ \sum_{i=\theta (n)+1}^{r}\dim_{F_{i}}LP(A_i, B_i) \cdot \frac{\deg f_{i}}{2}\nonumber\\
            &+ \sum_{i=r+1}^{r+s}\dim_{F_{i}}LP(A_i, B_i) \cdot \deg f_{i}.
        \end{align}
        We have $\dim_{F_{i}}\mC(A_i)=2m - \rk_{F_{i}}A_i$ and $\dim_{F_{i}}\mC(B_i)=2m - \rk_{F_{i}}B_i$. Let $\bullet \in \{A,B\}$ then 
$$\rk_{F_{i}}\bullet_{i} =  \begin{cases}
  0& \text{if $i \in [r+s]\setminus J_{1}^{\bullet} \cup J_{2}^{\bullet}$} \\
1 & \text{if $i \in J_{2}^{\bullet}$} \\
2 & \text{if $i \in J_{1}^{\bullet}$} \\
\end{cases}$$

and hence by Proposition \ref{prop:prop_dec}
\begin{equation}\label{eq:cases_dec}
    \dim_{F_{i}}LP(A_{i}, B_{i})= \begin{cases}
  2(2m-1)^{2}& \text{if $i \in \mathcal{J}_{1}$} \\
  2^2 (m-1)(2m-1)& \text{if $i \in \mathcal{J}_{2}$} \\
    2^{3}(m-1)^2& \text{if $i \in \mathcal{J}_{3}$} \\
      2^{3}m^{2}& \text{if $i \in \mathcal{J}_4$} \\
        2^{3}m(m-1)& \text{if $i \in \mathcal{J}_{5}$} \\
          2^{2}m(2m-1)& \text{if $i \in \mathcal{J}_{6}$}.\\
\end{cases}
\end{equation}
Now the result follows by combining  \eqref{eq:dec_formula} and \eqref{eq:cases_dec}.
        \end{proof}
\end{theorem}
\subsection{Induced codes}
Let $\ell$ be an arbitrary divisor of $n$ and $t <\ell$. Consider the two proper subgroups $D_{2(n/\ell)}= \langle a^{\ell}, ba^t \rangle$ and $C_{\ell}= \langle a^{n/\ell} \rangle$ of the dihedral group $D_{2n}$. Let $I$ be a left ideal of $\fq[C_{\ell}]$, then the 
code
$\mC := (\fqD)I$ is called $C_{\ell}$-induced. In \cite{zimmermann1994beitrage}
it was shown that if $I$ is an $[\ell, k,d]$ code, then $\mC$ is an $[2n, |\Gamma|k,d]$ code, where $\Gamma$ is the right transversal for $C_{\ell}$ in $D_{2n}.$
To distinguish between the different decompositions and the corresponding auxiliary code constructions, when referring to the algebra $\fq[C_{\ell}]$, we add $\;\hat{}\;$ to the notation. 

\begin{theorem} (cf.\ \cite[Theorem 6]{vedenev2020relationship}) \label{thm:induced_cyc}
    Let $\hat{x}^{\ell}-1 = (\prod_{i=1}^{\hat{r}}\hat{f}_{i}(\hat{x}))(\prod_{i=\hat{r}+1}^{\hat{r}+\hat{s}}\hat{f}_{i}(\hat{x})\hat{f}_{i}^{*}(\hat{x}))$ be the factorization of $\hat{x}^{\ell} -1$ into irreducible factors, $\hat{g} \mid \hat{x}^{\ell}-1$ and $\hat{\mC}_{\hat{g}} := (\hat{g})$.
    Let $\Omega : \fq[C_{\ell}] \hookrightarrow \fq[D_{2n}]$ be the embedding into $\fqD$ and let $\mC = (\fq[D_{2n}])\Omega(\hat{\mC}_{\hat{g}})$ be the induced code. Then $\mC$ is an $[2n, \frac{2kn}{\ell},d]$ code, where $k$ is the dimension of $\hat{\mC}_{\hat{g}}$ and $d$ its minimum distance. Moreover,
    \begin{equation*} \label{eq:cyclic_dec}
        \mC \cong \bigoplus_{j=1}^{r+s}B_{j}, \quad B_{j}= \begin{cases}R_j
 & \text{if $j \in J_1$} \\
 I_{j}
 & \text{if $j \in J_2$} \\ 
 0
 & \text{if $j \notin J_1 \cup J_{2}$} \\ 
\end{cases}, 
    \end{equation*}
    where 
    \begin{align*}
        &J_{1}= \{ j \in [r+s] : (f_{j}(x) \nmid \hat{g}(x^{n/\ell}) \wedge f_{j}^{*}(x) \nmid \hat{g}(x^{n/\ell}) \},\\
        &J_{2}= \{ j \in [r+s] \setminus [\theta (n)] :  \neg(f_{j}(x) \mid \hat{g}(x^{n/\ell}) \wedge f_{j}^{*}(x) \mid \hat{g}(x^{n/\ell}))\}.    
    \end{align*}
\end{theorem}

\begin{theorem} (cf.\ \cite[Theorem 7]{vedenev2020relationship}) \label{thm:dihedralinduced}
 Let $\hat{x}^{n/\ell}-1 = (\prod_{i=1}^{\hat{r}}\hat{f}_{i}(\hat{x}))(\prod_{i=\hat{r}+1}^{\hat{r}+\hat{s}}\hat{f}_{i}(\hat{x})\hat{f}_{i}^{*}(\hat{x}))$ be the factorization of $\hat{x}^{n/\ell} -1$ into irreducible factors, 
 let $\Omega : \fq[D_{2(n/\ell)}] \hookrightarrow \fq[D_{2n}]$ be the embedding into $\fqD$ and let $\hat{\mC} \subseteq \fq[D_{2(n/\ell)}]$ be a code such that
    $$\hat{\mC} \cong \bigoplus_{i=1}^{\hat{r}+\hat{s}}\hat{B}_{i},$$
    where, for $1 \le i \le \hat{r},\hat{B}_{i}=0$ or $\hat{B}_{i}= \hat{A}_{i}$. Let $\mC = (\fq[D_{2n}])\Omega(\hat{\mC})$ be the induced code, i.e., an $[2n, \ell k,d]$ code, where $k$ is the dimension of $\hat{\mC}_{\hat{g}}$ and $d$ its minimum distance. Moreover, suppose that 
    $$\mC \cong \bigoplus_{j=1}^{r+s}B_j.$$
 Then\footnote{Note that in \cite[Theorem~7]{vedenev2020relationship}, only divisibility conditions of the form $f_j(x)\mid \hat{f}_i(x^{\ell})$ are explicitly stated. The reciprocal polynomial cases $f_j^*(x)\mid \hat{f}_i(x^{\ell})$ are implicitly handled through a simplifying divisibility assumption (cf.\cite[Remark 6]{vedenev2020relationship}).}, for all $1 \le i \le \hat{r}+ \hat{s},$
   $$B_{j}= \begin{cases}R_j
 & \text{if $f_{j}(x) \mid \hat{f}_{i}(x^{\ell}) \vee f_{j}^{*}(x) \mid \hat{f}_{i}(x^{\ell})$\, \text{and}\,$\hat{B}_{i} = \hat{A}_{i},$} \\
   I_{j}& \text{if $f_{j}(x) \mid \hat{f}_{i}(x^{\ell}) \vee f_{j}^{*}(x) \mid \hat{f}_{i}(x^{\ell}) \, \text{and} \, \hat{B}_{i} = \hat{I}_{i},$} \\
  0
 & \text{else.} \\ 
\end{cases} $$  
\end{theorem}

\subsection{Distance bound}
We now determine the last missing parameter of our codes: the minimum distance. We will derive a lower bound on the minimum distance of the lifted product code, depending on the minimum distance of the codes related to the matrices used in the construction. 

\begin{notation} \label{notation:block_matrix}
For $c\in [\fq[G]]^m$ we consider the block vector$$\mathbf{c}:=[\mathbf{c_{1}}, \dots, \mathbf{c_m}] \in \fq^{|G|m},$$
where $\mathbf{c_{i}} \in \fq^{|G|}$ contains the coefficients of $c_i \in \fq[G]$.
\end{notation}

We need the following lemma in the proof of our main result in Theorem \ref{thm:distance}.

\begin{lemma}[See \cite{lin2023quantum}, Lemma 16] \label{lem:dist}
 
    Let $H$ be an abelian group and let $A \in \Mat_{m_{A}}(\fq[H]), B \in \Mat_{m_{B}}(\fq[H])$ with $\dim\mC(A) = \dim \mC (B)=0.$\footnote{Recall from Section 2.1 that $\mC (A)$ denotes the classical linear code with parity check matrix $A$.} Then the quasi-abelian code $LP(A,B)$ has zero dimension, i.e.,
    $$\dim LP(A,B)=0.$$
\end{lemma}
The following theorem is a variant of \cite[ Theorem 5]{kovalev2013quantum}, respectively \cite[Statement 12]{lin2023quantum}. For completeness we include a proof for our case.
\begin{theorem}\label{thm:distance}
Let $G_{A},G_{B} \subseteq D_{2n}$ be two proper subgroups such that $G_{AB}:= G_{A} \cap G_{B}$ is abelian and normal in both $G_{A},G_{B}$ and $[G_{A}:G_{AB}] \cdot [G_{B}:G_{AB}] \cdot |G_{AB}| =2n$.\footnote{This condition ensures that the code $LP(A,B)$ does not decompose into smaller, mutually disconnected subcodes associated with distinct double cosets of $G_{A}\backslash D_{2n}/G_{B}$. A similar condition was considered by Pryadko and Lin \cite{lin2023quantum} to guarantee code connectedness in their setting. The normality of $G_{AB}$ ensures that the replacements $a_{i,j} \mapsto {}^{\sharp}a_{i,j}, b_{i,j} \mapsto b_{i,j}^{\sharp}$ via left (and right) regular representations yield homomorphisms into matrix algebras over $\fq[G_{AB}]$.}
    Let $A \in \Mat_{m_{A}}(\fq[G_{A}]), B \in \Mat_{m_{B}}(\fq[G_{B}])$ and define 
    $$d_{0}:= \min\{d(\mC(A)),d(\mC(B)), d(\mC(A^\top)), d(\mC(B^\top))\}.$$ Then the minimum distance of the lifted product code $LP(A,B)$ satisfies
    $$D \ge \left\lfloor \frac{d_0}{|G_{AB}|} \right\rfloor.$$
\end{theorem}
    \begin{proof}
        Let $\ell_{A}:= [G_{A}:G_{AB}], \ell_{B}:= [G_{B}:G_{AB}]$. Moreover we replace the elements $a_{i,j}, b_{i,j}$ of $A,B$ by some square matrices. 
        More precisely we consider the left (respectively right) regular matrix representation with respect to a fixed basis of $\fq[G_{AB}]$ and define ${}^{\sharp}A:= [L_{G_{AB}}(a_{i,j})^\top]_{1 \le i,j \le m_{A}},$ $B^{\sharp}:= [R_{G_{AB}}(b_{i,j})]_{1 \le i,j \le m_{B}}$. Note that since the algebra $\fq[G_{AB}]$ is commutative, we do not need to distinguish between the left and right representations of $\fq[G_{AB}]$. Thus we use the bold notation from Definition \ref{def:matrix_rep} and define
        $$\mathbf{A} =\mathbf{{}^{\sharp}A} \otimes I_{\ell_{B}}, \quad \mathbf{B} = I_{\ell_{A}} \otimes \mathbf{B^{\sharp}}$$
        and the parity check matrices of $LP(A,B)$ by$$\mathbf{H_{X}} = [(\mathbf{{}^{\sharp}A} \otimes I_{\ell_{B}}) \otimes I_{m_{B}}, -I_{m_{A}} \otimes (I_{\ell_{A}} \otimes \mathbf{B^{\sharp}})], \quad \mathbf{H_{Z}} = [I_{n_{A}} \otimes(I_{\ell_{A}} \otimes \mathbf{B^{\sharp}}), -({}^{\sharp}\mathbf{A} \otimes I_{\ell_{B}}) \otimes I_{n_{B}}].$$

        Let $c \in \mC(H_{X})$ such that $w_{H}(c) < \lfloor d_{0} /|G_{AB}| \rfloor$. We define reduced matrices $${}^{\sharp}A[I_{A}]:=({}^{\sharp}a_{i,j})_{[m_{A}\ell_{A}]\times I_{A}}, \quad B^{\sharp}[I_{B}]:=(b^{\sharp}_{i,j})_{[m_{B}\ell_{B}]\times I_{B}},$$
        where $I_{A} \subseteq [m_{A}\ell_{A}], I_{B} \subseteq [m_{B}\ell_{B}]$ label the columns of ${}^{\sharp}A, B^{\sharp}$ incident to nonzero elements of $\mathbf{c}$ in $\mathbf{H_{X}}\mathbf{c}=0$. Let $\mI_A, \mI_B$ be the index sets of all columns in the corresponding $\mathbf{{}^{\sharp}A}, \mathbf{B^{\sharp}}$ and let $\mI = \mI_{A} \times [\ell_{B}m_{B}] \bigcup \mI_{B} \times [\ell_{A}m_{A}]$ be the labeling of all such columns in $\mathbf{H_{X}}.$
        Each element of $\fq[G_{AB}]$ corresponds to a block of size $|G_{AB}|$. Thus
        $$|\mathcal{I}_{\mu}| = |G_{AB}||I_{\mu}| \le |G_{AB}|w_{H}(c), \quad \mu \in \{A,B\}.$$
        Hence $\mathbf{{}^{\sharp}A}[I_{A}], \mathbf{B^{\sharp}}[I_{B}]$ have at most $d_{0}-1$ columns which implies that all columns in the parity check matrices are linearly independent. This gives $$\dim \mC({}^{\sharp}A[I_{A}]) = \dim \mC (B^{\sharp}[I_{B}])=0.$$
               Considering $LP({}^{\sharp}A[I_{A}], B^{\sharp}[I_{B}])$ with 
        \begin{align*}
            H_{X}[\mI] &= [{}^{\sharp}A[I_{A}] \otimes I_{m_{B}}, -I_{m_{A}} \otimes B^{\sharp}[I_{B}]]\\
             H_{Z}[\mI] &= [I_{|I_{A}|} \otimes {B^{\sharp}}^{\top}[I_{B}], {{}^{\sharp}A}^{\top}[I_{A}] \otimes I_{|I_{B}|}].
        \end{align*}
       Since $G_{AB}$ is abelian and both matrices ${}^{\sharp}A[I_{A}]$, $B^{\sharp}[I_{B}]$ are defined over $\fq[G_{AB}]$, Lemma \ref{lem:dist} gives $\dim LP({}^{\sharp}A[I_{A}], B^{\sharp}[I_{B}]) =0$. But this implies $\mC(H_{X}[\mI])= \mC(H_{Z}[\mI])^\perp$. Clearly, the reduced vector $c[\mI]$ belongs to $\mC(H_{X}[\mI])$ by construction. Hence $c[\mI] \in \mC(H_{Z}[\mI])^\perp$. Since $c$ can be obtained from $c[\mI]$ by extending it with zeroes on the positions $[2m_{A}m_{B}] \setminus \mI$ we have $c \in \mC(H_{Z})^\perp$. Similar arguments show that $c \in \mC(H_{Z})$ with $w_{H}(c) < \lfloor d_{0} /|G_{AB}| \rfloor$ belongs to $\mC(H_{X})^\perp.$

    \end{proof}

   \begin{remark}
     Note that with the theorem above, the best distance statements are achieved by taking two proper subgroups of $D_{2n}$ in the construction which have a trivial intersection. Therefore, particularly suitable are $C_{\ell}$ and $ D_{2(n/\ell)}$, where $\ell|n$. We illustrate the construction with these groups in the following example.
   \end{remark}

\begin{example} 
    We consider the dihedral group $D_{180}$ and its two proper subgroups $D_{20}$ (the dihedral group of order $20$) 
    and $C_{9}$ (the cyclic group of order $9$)

    intersecting in the trivial group $G_{AB}= \{e\}$, to present a nonabelian lifted product code with nontrivial guaranteed distance. In particular we consider the cyclic code
    $\hat{\mC}_{A}$ of length $9$ with generator polynomial
\begin{align*}
\hat{g}_{A}= (x^2 + x + 1)(x^6 + x^3 + 1)
\end{align*}

and let $\hat{e}_{A}=x-1$ be the (representative) check element, i.e., the generator of the dual code $\hat{\mC}_{A}^\perp$. 
The code $\hat{\mC}_{A}$ has minimum distance $9$ and we define $a=(a_{1}\hat{e}_{A}, \dots , a_{m}\hat{e}_{A})$ with $a_{i} \in \mathbb{F}_{11}[D_{180}]$ such that $a_{i}\mathbb{F}_{11}[D_{180}]+\hat{e}_{A}\mathbb{F}_{11}[D_{180}]=\mathbb{F}_{11}[D_{180}].$ An algorithmic idea to find suitable $a_i$ is given with Theorem \ref{thm:dimension}.
Using MAGMA we obtain
\begin{align*}
    \hat{g}_{A}(x^{10})= \prod_{i=1}^{20}\hat{g}_{i},
    \end{align*}
    where the irreducible factors can be found with Table \ref{tab:1}.
    
   We also consider the factorization
\begin{align*}
    (x^{90}-1) = \prod_{i=1}^{6}f_{i}\prod_{i=7}^{18}f_{i}^{*}f_{i},
\end{align*}
where each irreducible factor can be found with Table \ref{tab:2}.
Note that $r=6$ and $s=12$.
Using the notation of Theorem \ref{thm:induced_cyc} we obtain
  \begin{align*}
        J_{1}^{A}= \{ 1,2,7,8,9,10\}, \, J_{2}^{A}=\emptyset.
    \end{align*}
For the dihedral code we use the $[20,8,8]_{11}$-code $\hat{\mC}_{B} \subseteq \mathbb{F}_{11}[D_{90}]$ as obtained in \cite{vedenev2019codes}. More precisely, we consider $\hat{\mC}_{B} \subseteq \mathbb{F}_{11}[D_{180}]$ with decomposition
$$\hat{\rho}(\hat{\mC}_{B}) = \bigoplus_{i=1}^{6}\hat{B_{i}},$$
where
   $$\begin{array}{llllll}
    \hat{B_{1}} = \mathbb{F}_{11} \oplus \mathbb{F}_{11}, \hat{B_{2}} = 0 \oplus 0, 
    \hat{B_{3}} = I_{3}(1,-1),
    \hat{B_{4}} = M_{2}(\mathbb{F}_{11}[3]),  
      \hat{B_{5}} = 0, \hat{B_{6}} = 0.  \\
      \end{array} $$
      This code has a generator $\hat{g}_{B}$. The check element is the generator $\hat{g}_{B}^\perp$ of the dual $\hat{\mC}_{B}^\perp$. We let $\hat{e}_{B}= \hat{g}_{B}^\perp$ and define $\mathbf{b}= (b_1 \hat{e}_{B}, \dots, b_m \hat{e}_{B})$ with $b_{i} \in \mathbb{F}_{11}[D_{180}]$ such that $b_i\mathbb{F}_{11}[D_{180}]+\hat{e}_{B}\mathbb{F}_{11}[D_{180}] =\mathbb{F}_{11}[D_{180}]$. 
Moreover, we consider $A,B \in \text{Mat}_{m}(\mathbb{F}_{11}[D_{180}])$ such that all rows lie in $\mathbf{a}\mathbb{F}_{11}[D_{180}],$ respectively $\mathbf{b}\mathbb{F}_{11}[D_{180}]$. We then use the decomposition of $x^{10}-1$ which is 
$$(x^{10}-1)= \prod_{i=1}^{2}\hat{f}_{i}\prod_{i=3}^{6}\hat{f}_{i}^{*}\hat{f}_{i},$$
where each irreducible factor can be found with Table \ref{tab:3}.

   Factorizing $\hat{f}_{i}(x^{9})$ for $i \in [6]$ and applying Theorem \ref{thm:dihedralinduced} we have for the induced code $\mC_{B} = (\mathbb{F}_{11}[D_{2\cdot 90]}])\Omega(\hat{\mC}_{B})$ that

   $$\rho(\mC_{B})= \bigoplus_{i=1}^{18}B_i,$$
where   
   $$B_{i}= \begin{cases} A_{i}
 & \text{if $i \in \{1,3,5,8,12,17\},$} \\
 I_{i}
 & \text{if $i \in \{7,11,18\}$,} \\ 
 0
 & \text{else.} \\ 
\end{cases} $$ 
and hence
    \begin{align*}
        J_{1}^{B}= \{1,3,5,8,12,17\},\, J_{2}^{B}= \{7,11,18\}.
    \end{align*}
 We obtain from Theorem \ref{thm:dimension} that
    \begin{align*}
        &\dim_{\mathbb{F}_{11}} LP(A,B) = 4(m-1)(2m-1)+12(m-1)^{2}+160m^2 + 116m(m-1)+32m(2m-1)
    \end{align*}
and show the possible parameter sets (for $m=1,\dots,5$) of these codes in Table \ref{table:1}: 
    \begin{table}[htb!]
\centering
\begin{tabular}{|c|c|c|} 
\hline
$m$ & $[[N,K,D]]_{11}$ code & $K/N$ \\ [0.5ex] 
\hline\hline
1 & $[[360,192,8]]_{11}$&0.53 \\ 
\hline
2 & $[[1440,1088,8]]_{11}$ & 0.76 \\
\hline
3 & $[[3240,2704,8]]_{11}$ & 0.83 \\
\hline
4 & $[[5760,5040,8]]_{11}$ & 0.88 \\
\hline
5 & $[[9000,8096,8]]_{11}$ & 0.9 \\
\hline
\end{tabular}
\caption{Dihedral lifted product codes obtained from a $[9,1,9]_{11}$-cyclic code and a $[20,8,8]_{11}$-dihedral code. The third column describes the rate of the code.}
\label{table:1}
\end{table}

\end{example}

\section{Conclusion} \label{sec:final}

In this paper we concentrated on nonabelian group code constructions.
Although the existence and construction of good quantum CSS codes over various fields, including $\mathbb{F}_{11}$, have been  extensively explored and demonstrated in \cite{grassl2015quantum}, our codes offer distinct advantages due to their MDPC structure. One significant advantage of MDPC codes is their decodability via graph-based decoders. Recent advancements, particularly iterative belief propagation, see \cite{leverrier2023efficient}, and neural-network-assisted decoding techniques, see \cite{gong2024graph}, demonstrate robust and efficient decoding performance for quantum LDPC codes. These decoders exploit the sparse and  structured nature of LDPC parity-check matrices, significantly reducing computational complexity compared to generic decoding approaches and can be adapted for MDPC codes straight-forwardly.

Notably, among the 2-block codes proposed in \cite{wang2023abelian}, our codes appear to be the first nonabelian quantum MDPC codes, paving the way to the development of a quantum McEliece public key cryptosystem (as in \cite{fujita2012quantum}) based on quantum MDPC codes.

For future work it would be interesting to consider and analyze decoding algorithms for these dihedral codes; for example, via the generalized discrete Fourier transforms, or the Morita correspondence between $\fq[x]/(x^m -1)$ submodules and left ideals in $\Mat_{2}(\fq)[x]/(x^m -1)$ as discussed in \cite{barbier2012quasi,borello2021dihedral}.

\bibliographystyle{amsplain}
 \bibliography{biblio.bib}

\newpage
\appendix
\section{Appendix}

\begin{table}[htb!] 
    \centering
    \renewcommand{\arraystretch}{1.2} 
    \begin{tabular}{|c| l l l l l l l|} 
        \hline
        degree & 0 & 1 & 2 & 3 & 4 & 5 & 6 \\
        \hline
        $\hat{g}_{1}$  & 1  & 1  & 1  & 0  & 0  & 0  & 0  \\
        \hline
        $\hat{g}_{2}$  & 4 & 2 & 1 & 0 & 0 & 0 & 0 \\
       \hline
        $\hat{g}_{3}$  & 9 & 3 & 1 & 0 & 0 & 0 & 0 \\
        \hline
        $\hat{g}_{4}$  & 5 & 4 & 1 & 0 & 0 & 0 & 1 \\
         \hline
        $\hat{g}_{5}$  & 3 & 5 & 1 & 10 & 0 & 0 & 1 \\
         \hline
        $\hat{g}_{6}$  & 3 & 6 & 1 & 0 & 0 & 0 & 0 \\
          \hline
       $\hat{g}_{7}$  & 5 & 7 & 1 & 0 & 0 & 0 & 0 \\
          \hline
        $\hat{g}_{8}$  & 9 & 8 & 1 & 0 & 0 & 0 & 0 \\
          \hline
       $\hat{g}_{9}$  & 4 & 9 & 1 & 0 & 0 & 0 & 0 \\
           \hline
        $\hat{g}_{10}$  & 1 & 10 & 1 & 0 & 0 & 0 & 0 \\
          \hline
       $\hat{g}_{11}$  & 1 & 0 & 0 & 1 & 0 & 0 & 1 \\
          \hline
       $\hat{g}_{12}$  & 4 & 0 & 0 & 2 & 0 & 0 & 1 \\
          \hline
           $\hat{g}_{13}$  & 9 & 0 & 0 & 3 & 0 & 0 & 1 \\
          \hline
           $\hat{g}_{14}$  & 5 & 0 & 0 & 4 & 0 & 0 & 1 \\
          \hline
           $\hat{g}_{15}$  & 3 & 0 & 0 & 5 & 0 & 0 & 1 \\
          \hline
           $\hat{g}_{16}$  & 3 & 0 & 0 & 6 & 0 & 0 & 1 \\
          \hline
           $\hat{g}_{17}$  & 5 & 0 & 0 & 7 & 0 & 0 & 1 \\
          \hline
           $\hat{g}_{18}$  & 9 & 0 & 0 & 8 & 0 & 0 & 1 \\
          \hline
           $\hat{g}_{19}$  & 4 & 0 & 0 & 9 & 0 & 0 & 1 \\
          \hline
           $\hat{g}_{20}$  & 1 & 0 & 0 & 10 & 0 & 0 & 1 \\
          \hline
    \end{tabular}
    \caption{Irreducible factors of $\hat{g}_{A}(x^{10}) \in \mathbb{F}_{11}[x],$ where $\hat{g}_{A}(x)=(x^{2}+x+1)(x^{6}+x^{3}+1).$ The first column lists the irreducible factors denoted by $\hat{g}_i$, the remaining columns list the coefficients of the monomials of each factor.}
  
    \label{tab:1}
\end{table}

\begin{table}[h] 
    \centering
    \renewcommand{\arraystretch}{1.2} 
    \begin{tabular}{|c| l l l l l l l|} 
        \hline
        degree & 0 & 1 & 2 & 3 & 4 & 5 & 6 \\
        \hline
        $f_{1}$  & 10  & 1  & 0  & 0  & 0  & 0  & 0  \\
        \hline
        $f_{2}$  & 1 & 1 & 0 & 0 & 0 & 0 & 0 \\
       \hline
        $f_{3}$  & 1 & 1 & 1 & 0 & 0 & 0 & 0 \\
        \hline
        $f_{4}$  & 1 & 10 & 1 & 0 & 0 & 0 & 0 \\
         \hline
        $f_{5}$  & 1 & 0 & 0 & 1 & 0 & 0 & 1 \\
         \hline
        $f_{6}$  & 1 & 0 & 0 & 10 & 0 & 0 & 1 \\
          \hline
        $f_{7}$  & 9 & 1 & 0 & 0 & 0 & 0 & 0 \\
       $f_{7}^{*}$  & 5 & 1 & 0 & 0 & 0 & 0 & 0 \\
           \hline
        $f_{8}$  & 7 & 1 & 0 & 0 & 0 & 0 & 0 \\
       $f_{8}^{*}$  & 8 & 1 & 0 & 0 & 0 & 0 & 0 \\
          \hline
        $f_{9}$  & 3 & 1 & 0 & 0 & 0 & 0 & 0 \\
       $f_{9}^{*}$  & 4 & 1 & 0 & 0 & 0 & 0 & 0 \\
          \hline
        $f_{10}$  & 2 & 1 & 0 & 0 & 0 & 0 & 0 \\
       $f_{10}^{*}$  & 6 & 1 & 0 & 0 & 0 & 0 & 0 \\
            \hline
        $f_{11}$  & 4 & 2 & 1 & 0 & 0 & 0 & 0 \\
       $f_{11}^{*}$  & 3 & 6 & 1 & 0 & 0 & 0 & 0 \\
             \hline
        $f_{12}$  & 9 & 3 & 1 & 0 & 0 & 0 & 0 \\
       $f_{12}^{*}$  & 5 & 4 & 1 & 0 & 0 & 0 & 0 \\
             \hline
        $f_{13}$   & 5 & 7 & 1 & 0 & 0 & 0 & 0\\
       $f_{13}^{*}$  & 9 & 8 & 1 & 0 & 0 & 0 & 0 \\
              \hline
        $f_{14}$  & 3 & 5 & 1 & 0 & 0 & 0 & 0 \\
       $f_{14}^{*}$  & 4 & 9 & 1 & 0 & 0 & 0 & 0 \\
              \hline
        $f_{15}$  & 4 & 0 & 0 & 2 & 0 & 0 & 1 \\
       $f_{15}^{*}$  & 3 & 0 & 0 & 6 & 0 & 0 & 1 \\
              \hline
        $f_{16}$  & 9 & 0 & 0 & 3 & 0 & 0 & 1 \\
       $f_{16}^{*}$  & 5 & 0 & 0 & 4 & 0 & 0 & 1 \\
              \hline
        $f_{17}$  & 3 & 0 & 0 & 5 & 0 & 0 & 1 \\
       $f_{17}^{*}$  & 4 & 0 & 0 & 9 & 0 & 0 & 1 \\
       \hline
       $f_{18}$  & 5 & 0 & 0 & 7 & 0 & 0 & 1 \\
       $f_{18}^{*}$  & 9 & 0 & 0 & 8 & 0 & 0 & 1 \\
       \hline
    \end{tabular}
    \caption{Irreducible factors of the polynomial $x^{90}-1 \in \mathbb{F}_{11}[x].$ The first column lists the irreducible factors, the remaining columns list the coefficients of the monomials of each factor.}
    \label{tab:2}
\end{table}

\begin{table}[h] 
    \centering
    \renewcommand{\arraystretch}{1.2} 
    \begin{tabular}{|c| l l|} 
        \hline
        degree & 0 & 1  \\
        \hline
        $\hat{f}_{1}$  & -1  & 1  \\
        \hline
        $\hat{f}_{2}$  & 1 & 1  \\
       \hline
        $\hat{f}_{3}$  & -2 & 1 \\
           $\hat{f}_{3}^{*}$  & -6 & 1 \\
        \hline
        $\hat{f}_{4}$  & -3 & 1  \\
        $\hat{f}_{4}^{*}$  & -4 & 1 \\
         \hline
        $\hat{f}_{5}$  & -7 & 1 \\
          $\hat{f}_{5}^{*}$  & -8 & 1\\
         \hline
        $\hat{f}_{6}$  & -9 & 1\\
          $\hat{f}_{6}^{*}$  & -5 & 1\\
          \hline
    \end{tabular}
    \caption{Irreducible factors of the polynomial $x^{10}-1 \in \mathbb{F}_{11}[x].$ The first column lists the irreducible factors, the remaining columns list the coefficients of the monomials of each factor.}
    \label{tab:3}
\end{table}
\end{document}